\documentclass{article}

\usepackage[english]{babel}

\usepackage[letterpaper,top=2cm,bottom=2cm,left=3cm,right=3cm,marginparwidth=1.75cm]{geometry}

\usepackage{amsfonts,amsmath,amssymb,amsthm,amscd}
\usepackage{graphicx}
\usepackage{mathtools}
\usepackage{bm}%, bbm}
\usepackage{booktabs}
\usepackage[colorlinks=true, allcolors=blue]{hyperref}
\usepackage{pifont}% http://ctan.org/pkg/pifont
\newcommand{\cmark}{\ding{51}}%
\newcommand{\xmark}{\ding{55}}%
\usepackage{natbib}
\usepackage{authblk}
\usepackage{dsfont}

\newtheorem{theorem}{Theorem}[section]
\newtheorem{proposition}{Proposition}[section]

\theoremstyle{remark}
\newtheorem*{remark}{Remark}

\newcommand{\dd}{\mathrm{d}}
\newcommand{\R}{\mathbb{R}}

\newcommand{\Z}{\mathbb{Z}}

\newcommand{\E}{\mathbb{E}}
\newcommand{\Var}{\mathrm{Var}}
\newcommand{\Cov}{\mathrm{Cov}}

\newcommand{\1}{\mathds{1}}

\newcommand{\bigO}{\mathcal{O}}
\newcommand{\iu}{\mathrm{i}}

\DeclareMathOperator*{\argmin}{arg\,min}

\author[1]{Miguel Martinez Herrera}
\author[2]{Felix Cheysson}

\affil[1]{Institut Pasteur, Université Paris Cité, CNRS UMR3738, Zebrafish Neurogenetics Unit, F-75015 Paris, France. Email: miguel.martinez-herrera@pasteur.fr}
\affil[2]{LAMA, CNRS UMR 8050, Université Gustave Eiffel, 77420 Champs-sur-Marne, France. Email: felix.cheysson@univ-eiffel.fr}

\title{Ridge-penalised spectral least-squares estimation for point processes}

\begin{document}
\maketitle

\begin{abstract}
Penalised estimation methods for point processes usually rely on a large amount of independent repetitions for cross-validation purposes.
    However, in the case of a single realisation of the process, existing cross-validation methods may be impractical depending on the chosen model.
    To overcome this issue, this paper presents a Ridge-penalised spectral least-squares estimation method for second-order stationary point processes.
    This is achieved through two novel approaches: a $p$-thinning-based cross-validation method to tune the penalisation parameter, relying on the spectral representation of the process; and the introduction of a spectral least-squares contrast based around the asymptotic properties of the periodogram of the sample.
    The proposed method is then illustrated by a simulation study on linear Hawkes processes in the context of parametric estimation, highlighting its performances against more traditional approaches, specifically when working with short observation windows.
\end{abstract}

\section{Introduction}

As point data become more readily available, point process models are increasingly used across various fields \citep{Illian2007, Baddeley2016}. 
Common estimation methods for point processes typically involve maximum likelihood \citep{Ogata1978}, spectral techniques like Whittle likelihood \citep{Adamopoulos1976, Yang2026}, or more \textit{ad-hoc} approaches, such as minimum contrast methods \citep{Biscio2017, Diggle2013} or least squares contrast \citep{Reynaud-Bouret2010}.

Linear Hawkes processes, introduced by \cite{Hawkes1971}, form a class of point processes in which the occurrence of an event increases the probability of subsequent events. 
Initially applied in seismology \citep{Adamopoulos1976}, their use has expanded to fields including but not limited to genomics \citep{Reynaud-Bouret2010}, epidemiology \citep{Meyer2012}, neurology \citep{Reynaud-Bouret2014}, and finance \citep{Bacry2015}.

In this article, we are interested in Ridge-penalised estimation for Hawkes processes.
Despite their increasing relevance, penalised estimation techniques for Hawkes processes remain limited in the literature, with existing works focusing on Lasso penalisation \citep{Reynaud-Bouret2014, Bacry2020, Denis2025}. 
A critical step in these techniques usually involves the choice of the penalisation parameter: in \cite{Reynaud-Bouret2014} and \cite{Bacry2020}, the penalisation parameter is determined using asymptotic results, while \cite{Denis2025} selects it via a specific model criterion (namely EBIC). 
These works usually assume that many independent repetitions of the Hawkes process are observed, or that the observation window is very large.
There is, however, a significant lack of cross-validation methods for applying penalised estimation techniques---such as Lasso and Ridge---in point processes.

One notable contribution is \cite{Cronie2024}, who introduces a $p$-thinning-based cross-validation method for point processes, leveraging the Papangelou conditional intensity. 
But this intensity is often untractable for many point process families, or their thinned versions.
To address this issue, we propose to leverage their often explicit spectral measure, allowing our approach to be applicable where previous methods may struggle.
Indeed, there has been a recent interest in spectral approaches to estimation for point processes \citep{Cheysson2022, Yang2026, Bonnet2024}, often making up for the lack of tractable conditional intensity in specific models.

This paper proposes a spectral approach to $p$-thinning-based cross-validation, complementing the approach of \cite{Cronie2024}. 
This cross-validation technique is then used to select the Ridge-penalisation parameter in a novel spectral framework to least-squares estimation.
Although our method is illustrated on linear Hawkes processes, it is broadly applicable to any spatial, second-order stationary point process.

The paper is organized as follows: after recalling some notations in Section \ref{sec:notations}, Section \ref{sec:contrast} introduces a spectral least-squares contrast and estimator in a general framework, based on the periodogram of the point process. 
In Section \ref{sec:random_subsampling}, we describe the $p$-thinning subsampling cross-validation method, in the context of Ridge-penalised spectral least-squares estimation.
Finally, Section \ref{sec:hawkes_and_simulations} presents a simulation study that illustrates the performance of our proposed estimation technique for the linear Hawkes process.

\section{Notation}\label{sec:notations}

Let $N$ be a second-order stationary point process on $\R$ with mean intensity $m^*$ and spectral measure $\Gamma^*$, assumed to be absolutely continuous with respect to the Lebesgue measure and with density $f^*$.
Formally, the spectral measure $\Gamma^*$ of a point process $N$ is the unique measure on $\mathcal B(\R)$ such that, 
for any $\varphi\in L^2(\mathbb{R})$, \citep[Proposition 8.2.I]{DaleyV1}: 
\begin{align}\label{eq:bartlett_variance}
    \Var\left[\int_{\R}{\varphi(x)\,N(\dd x)}\right] 
    &= \int_{\R}{|\widetilde\varphi(\nu)|^2\,\Gamma^*(\dd \nu)}\\ 
    &= \int_{\R}{|\widetilde\varphi(\nu)|^2\,f^*(\nu)\,\dd \nu},\notag
\end{align}
with $\widetilde\varphi(\nu) = \int_\R e^{-2i\pi \nu x} \varphi(x) \dd x$ the Fourier transform of $f$.
Existence of such a measure is established for any stationary point process \citep[Proposition 8.2.I.(a)]{DaleyV1}.
Denoting $f_0^*$ the Fourier transform of the factorial covariance measure of $N$, the spectral density can always be decomposed such that \citep[Equation 8.2.4]{DaleyV1}
\[f^*(\nu) = m^* + f_0^*(\nu)\,,\]
for all $\nu\in\R$, with $f_0^* \in L^2(\R)$ its compensated spectral density.

\section{A spectral least-squares estimator}\label{sec:contrast}

Let us provide a least-squares contrast estimator for the stationary point process $N$, in a similar manner as in \cite{Reynaud-Bouret2010}.
Suppose we observe a realisation of the process $N$ on the interval $[0, T]$ and want to estimate its mean intensity and spectral measure $f = (m, f_0)$.
Let $D$ be a fixed compact on $\R$, of the form $D = [-A, A]$, and define
\begin{equation*}
    \mathcal{F} = \{f=(m, f_0) \colon f_0\in L^2(D),
    \|f\|_{\mathcal{F}}^2 := m^2 + \|f_0\|_{L^2(D)}^2 < +\infty\}
\end{equation*}
as the space of candidate functions for the spectral density of $N$. 
Note that $f^* = (m^*, f_0^*)$, the (true) spectral density of $N$, belongs to $\mathcal F$.
It is clear that $\|\cdot\|_\mathcal{F}$ is a norm.

Measuring the distance between a candidate $f\in\mathcal{F}$ and $f^*$,
\begin{equation*}
\|f - f^*\|_{\mathcal F}^2 = (m - m^*)^2 + \|f_0 - f_0^*\|_{L^2(D)}^2,
\end{equation*}
it is straightforward to see that it is minimised for $f = f^*$.
Let us focus on $\|f_0 - f_0^*\|_{L^2(D)}^2$, which can be rewritten as follows
\begin{align}\label{eqn:L_2(D)}
    \|f_0 - f_0^*\|_{L^2(D)}^2 &= \int_{D}{\lvert f_0(\nu) - f_0^*(\nu)\rvert^2\,\dd \nu}\\
    &= \int_{D}{f_0^*}^2(\nu) \, \dd \nu + \int_{D}f_0^2(\nu) \, \dd \nu 
    - 2\int_{D} f_0(\nu)f_0^*(\nu) \,\dd \nu\notag.
\end{align}
Minimising the last expression with respect to $f$ is equivalent to minimising the last two terms, and so we can define:
\begin{equation*}
    \gamma^*(f) = (m - m^*)^2 + \int_{D}f_0^2(\nu)\,\dd \nu
    - 2\int_{D} f_0(\nu)f_0^*(\nu) \,\dd \nu.
\end{equation*}
By replacing $f_0^* = f^* - m^*$, we obtain:
\begin{equation}\label{eq:function_to_minimise}
        \gamma^*(f) = (m-m^*)^2 + \int_{D} f_0^2(\nu) \, \dd\nu
        - 2 \int_{D}{f_0(\nu) \, {f^*}(\nu) \,\dd \nu} + 2 m^*\int_{D}{f_0(\nu) \, \dd \nu}.
\end{equation}
Although this last quantity is still unobservable, we may plug in classical estimators for $m^*$ and $f_0^*$.
Let $\widehat m = N_T /T$ be the estimated average intensity and $\widehat I_T$ the periodogram of the process $N$ on $[0, T]$, defined as
\begin{equation*}
    \widehat I_T(\nu) = \frac{1}{T} \left| \int_0^T e^{-2\pi \iu \nu t} \bigl( N(\dd t) - \widehat m \dd t \bigr) \right|^2. 
\end{equation*}
Under second-order stationarity \citep{Yang2026}, for any $\nu\in \R \setminus \{0\}$,
\begin{equation*}
    \lim_{T\to\infty} \E[\widehat I_T(\nu)] = f^*(\nu).
\end{equation*}
We can then define the following observable quantity,
\begin{align}\label{eq:contrast}
    \begin{split}
        \gamma(f) &= (m-\widehat m)^2 + \int_{D} f_0^2(\nu) \, \dd\nu
        - 2 \int_{D}{f_0(\nu) \, \widehat I_T(\nu) \, \dd \nu} + 2 \widehat m\int_{D}{f_0(\nu) \, \dd \nu}
    \end{split}\notag\\
    &=: (m-\widehat m)^2 + \gamma_0(f_0).
\end{align}
Leveraging the Theorem 4.1 of \cite{Yang2026} yields that $\gamma$ is, asymptotically, a contrast.

\begin{proposition}\label{prop:contrast}
    Suppose that the process $N$ is second-order stationary with integrable second-order reduced cumulant intensity, and that the density $f_0$ is twice differentiable with bounded derivatives.
    Then,
    \begin{equation*}
        \lim_{T\to\infty}\E[\gamma(f)] = \gamma^*(f).
    \end{equation*}
    Further assume that the process $N$ is fourth-order stationary with integrable fourth-order reduced cumulant intensity, and that its fourth-order spectral cumulant is twice differentiable with bounded partial derivatives, then 
    \begin{equation*}
        \underset{T\to\infty}{\mathrm{plim}} \, \gamma(f) = \gamma^*(f).
    \end{equation*}  
\end{proposition}

\begin{remark}
    Assumptions regarding the fourth-order properties of the process are needed to ensure that the variance of the integrated periodogram, which appears for example in Equation \eqref{eq:contrast}, converges to zero when $T \to \infty$.
\end{remark}

Since $\gamma^*(f)$ is minimal for $f = f^*$ and by building on the previous proposition, then from Equation~\eqref{eq:contrast}, we may define the spectral least-squares estimator $\widehat f = (\widehat m, \widehat f_0)$ by
\begin{equation}\label{eqn:estimator}
    \widehat m = \frac{N_T}{T} \quad \text{and}\quad \widehat f_0\in \argmin_{f_0\in L^2(D)} \gamma_0(f_0).
\end{equation}

\begin{remark}
As $\gamma^*$ derives from the $L^2$ norm given in Equation \eqref{eqn:L_2(D)}, we may consider instead as an estimator $\widehat f^\perp = (\widehat m, \widehat f_0^\perp)$, with $\widehat f_0^\perp$ the $L^2(D)$-projection of $\widehat I_T - \widehat m$.
Then, $\widehat f_0^\perp$ also minimises $\gamma_0$ and $\widehat f = \widehat f^\perp$. 
Simulation studies in Section \ref{sec:performance_metrics} will illustrate this.
However, theoretical guarantees cannot be easily obtained for $\widehat f^\perp$ as it involves an additional term, $\int_D \widehat I_T^2(\nu) \dd\nu$, for which the asymptotic behaviour has not yet been studied in the literature, hence the need to work with the contrast $\gamma$.
\end{remark}

\paragraph{The spatial tapered case.}
Let us note that, while we presented, for clarity and ease of reading, the framework for a spectral least-squares estimator for a non-tapered periodogram and in a temporal setting, this framework can easily be extended to the tapered case and in $\R^d$, and the previous results still hold: we refer readers to the notations and adjustments given in \cite{Yang2026}.
In this context, we present very briefly the spectral least-squares estimator.

For a second-order stationary process $N$ on $\R^d$ observed within a compact domain $D_T = [-T/2, T/2]^d$ and a non-negative data taper $h$ on $\R^d$ with compact support $[-1/2, 1/2]^d$, let $\widehat m = N(D_T) / \lvert D_T \rvert$ and, for $\nu \in \R^d$, define
\begin{equation*}
    \widehat I_T(\nu) = \lvert D_T \rvert^{-1} H_{h,2}^{-1}\left| \int_{D_T} h(x/T) e^{-2\pi \iu x \cdot \nu} \bigl( N(\dd x) - \widehat m \dd x \bigr) \right|^2,
\end{equation*}
with $H_{h,2} = \int_{[-1/2, 1/2]^d} h(x)^2\dd x$.
Then the estimator $\widehat f = (\widehat m, \widehat f_0)$, with $\widehat f_0 \in \argmin_{f_0\in L^2(D)} \gamma_0(f_0)$,
\begin{equation*}
    \gamma_0(f) = \int_{D} f_0^2(\nu) \, \dd\nu - 2 \int_{D}{f_0(\nu) \, \widehat I_T(\nu) \, \dd \nu}
    + 2 \widehat m\int_{D}{f_0(\nu) \, \dd \nu},
\end{equation*}
and $D = [-A, A]^d$ a compact, is the spectral least-squares estimator for $f$.

\section{Random subsampling}\label{sec:random_subsampling}
\subsection{Spectral measure of a \textit{p}-thinned process}

For any $p \in (0,1)$, we introduce a $p$-thinned version of the process $N$, denoted $N_p$ and defined for any $B \in \mathcal{B}(\R)$ as
\[N_p(B) = \sum_{k\in\Z}{\1_{T_k \in B} Z_k}\,,\]
where $(T_k)_{k \in \Z}$ denotes the atoms of $N$, and $(Z_k)_{k\in\Z}$ is an i.i.d.\ collection of Bernoulli random variables of parameter $p$.
In essence, the atoms of $N_p$ correspond to a subset of $(T_k)_{k\in\Z}$ where each point is erased independently of the others with probability $1-p$.  
Leveraging the study of spectral quantities on marked point processes \citep{Bremaud2002, Bremaud2005}, we can establish an explicit expression of the spectral density of the $p$-thinning $N_p$, as established hereafter.
\begin{proposition}\label{prop:spectral_thinning}
    Let $N$ be a stationary point process with mean intensity $m^*$ and spectral density function $f^* = m^* + f_0^*$, and for any $p \in (0,1)$, let $N_p$ be a $p$-thinning of $N$.
    Then, $N_p$ has mean intensity $p m^*$ and admits a spectral density function, denoted $f_p^*$, given for all $\nu \in \R$ by
    \begin{align}\label{eq:spectral_thinning}
        f_p^*(\nu) &= p^2 f^*(\nu) + p(1-p)m^*\\
        &= p^2 f_0^*(\nu) + p m^*.\nonumber
    \end{align}
\end{proposition}

\begin{proof}
    The proof is given in Appendix~\ref{app:proof_spectral_thinning}.
\end{proof}

Let us remark that Equation~\eqref{eq:spectral_thinning} can be found in \cite[Example 8.3(b)]{DaleyV1}, where they obtain this result by identifying first- and second-order properties of a bivariate point process. 
The proof presented in the appendix is an alternative way of establishing this expression, illustrating the usefulness of the spectral theory of point processes.

\subsection{\textit{p}-thinning as a subsampling method}\label{sec:p_thinning_as_subsampling}
In this section, we assume that we are provided with a single observation on a relatively short interval $[0,T]$ of a process $N$ belonging to the parametric model defined by
\begin{equation*}
    \mathcal{F}_\Theta = \{m\colon 0 < m < \infty\} \times \{ f_{0, \theta} \colon \theta \in \Theta, f_{0, \theta} \in L^2(D)\},
\end{equation*}
with $\Theta \subset \R^p$ a compact subset.
To improve estimation, we consider a penalised, namely Ridge, version of the contrast given in Equation \eqref{eq:contrast}: for any $\kappa > 0$, consider the Ridge estimator given by
\begin{equation*}
    \widehat m = \frac{N_T}{T}, \quad \text{and} \quad \widehat\theta^{(\kappa)} = \argmin_{\theta \in \Theta} \gamma_0(f_{0, \theta}) + \kappa \lVert \theta \rVert_2^2.
\end{equation*}
A common practice to select the penalisation constant $\kappa$ is to consider a grid of acceptable values for $\kappa$ and, for each, estimate the performance of the estimated model on an independent testing dataset.
When only one dataset is available, it is usually repeatedly split into multiple training and testing sets using e.g.\ bootstrapping or subsampling.
In our setting, subsampling is achieved through $p$-thinning, where retained points form the training process $N^{(p)}$ while rejected points form the testing process $N^{(\bar p)}$.

Formally, let $p \in (0,1)$, $\kappa > 0$, and let $N^{(p)}$ denote a $p$-thinned version of the process $N$ as in previous subsection.
Since the mean intensity and compensated spectral density of $N^{(p)}$ are given by $pm^*$ and $p^2 f_{0, \theta^*}$ respectively, we may define as an estimator of $\theta^*$ the quantity
\begin{equation*}
    \qquad \widetilde \theta^{(p, \kappa)} = \argmin_{\theta \in \Theta} \bigl\{ \gamma_0^{(p)}(f_{0, \theta}) + \kappa \lVert \theta \rVert_2^2 \bigr\},
\end{equation*}
where $\gamma_0^{(p)}$ is defined by
\begin{equation}\label{eqn:pcontrast}
    \gamma_0^{(p)}(f_{0, \theta}) = \int_{D} f_{0,\theta}^2(\nu) \, \dd\nu - 2 \int_{D}{f_{0,\theta}(\nu) \, \widetilde I_T^{(p)}(\nu) \, \dd \nu}
    + 2 \widehat m \int_{D}{f_{0, \theta}(\nu) \, \dd \nu},
\end{equation}
with
\begin{equation*}
    \widetilde I_T^{(p)}(\nu) = \frac{ \widehat I_T^{(p)} (\nu) - p(1-p) \widehat m}{p^2},
\end{equation*}
and
\begin{equation*}
    \widehat I_T^{(p)}(\nu) = \frac{1}{T} \left| \int_0^T e^{-2\pi \iu \nu t} \bigl( N^{(p)}(\dd t) - p \widehat m \dd t \bigr) \right|^2. 
\end{equation*}
The performance of the estimation for a given couple $(p, \kappa)$ is then given by computing the least-squares contrast on the testing set $N^{(\bar p)}$ consisting of all points in $N$ not retained in $N^{(p)}$, that is $N^{(\bar p)} = N - N^{(p)}$.
This is given by the following quantity,
\begin{equation*}
    r^{(p, \kappa)} \coloneqq \gamma_0^{(\bar p)} \bigl(f_{0, \widetilde \theta^{(p, \kappa)}} \bigr),
\end{equation*}
with $\gamma_0^{(\bar p)}$ defined similarly to $\gamma_0^{(p)}$ but with 
\begin{equation*}
    \widetilde I_T^{(\bar p)}(\nu) = \frac{ \widehat I_T^{(\bar p)} (\nu) - p(1-p) \widehat m}{(1-p)^2},
\end{equation*}
and 
\begin{equation*}
    \widehat I_T^{(\bar p)}(\nu) = \frac{1}{T} \left| \int_0^T e^{-2\pi \iu \nu t} \bigl( N^{(\bar p)}(\dd t) - (1 - p) \widehat m \dd t \bigr) \right|^2. 
\end{equation*}
Here, note that $\widehat I_T^{(p)}$ denotes the periodogram of $N^{(p)}$, i.e.\ the plug-in estimator of $f_p^* = pm^* + p^2 f_{0, \theta^*}$, while $\widetilde I_T^{(p)}$ denotes a rescaled version of $\widehat I_T^{(p)}$, and an estimator of $f^* = m^* + f_{0, \theta^*}$ (and likewise for $N^{(\bar p)}$).
This rescaling (rather than working on the scale of the $p$-thinned process) is important, since we aim to compare the performances of the estimator for different values of $p$.

In practice, we consider $n > 0$ independent repetitions $N_j^{(p)}$, $1 \le j \le n$, of the $p$-thinning from the original process $N$ for a grid of acceptable $(p, \kappa)$, yielding multiple estimators $\widetilde \theta_j^{(p, \kappa)}$ and testing errors $r_j^{(p, \kappa)}$.
Then, selecting
\begin{equation*}
    (\widehat p, \widehat \kappa) = \argmin_{(p, \kappa)} \frac{1}{n} \sum_{j=1}^n r_j^{(p, \kappa)} ,
\end{equation*}
we obtain the spectral Ridge estimator by averaging the estimators on each subsampled process, 
\begin{equation*}
    \widehat \theta^{(\widehat p, \widehat \kappa)} = \frac{1}{n} \sum_{j=1}^n \widetilde \theta_j^{(\widehat p, \widehat \kappa)}.
\end{equation*}

\section{Application to the Hawkes process}\label{sec:hawkes_and_simulations}

\subsection{Theoretical framework}

The linear Hawkes process is a doubly stochastic point process on $\R$ exhibiting self-excitation and clustering \citep{Hawkes1971, Hawkes1974}.
It is defined by its conditional intensity function, which describes the rate of arrival of points given the past and takes the form
\begin{equation*}
    \lambda_t = \mu + \int_{-\infty}^t h(t-u) N(\dd u),
\end{equation*}
where $\mu > 0$ is the baseline intensity, and $h: \R_+ \to \R_+$, the reproduction function of the process, is a measurable nonnegative function such that $\lVert h \rVert_1< 1$.
Its spectral density is well-known \citep[Example 8.2(e)]{DaleyV1} and given by
\begin{equation*}
    f(\nu) = m \bigl\lvert 1 - \widetilde h(\nu) \bigr\rvert^{-2},
\end{equation*}
with $m = \mu / (1 - \lVert h \rVert_1)$ the stationary intensity of the process, and $\widetilde h$ the Fourier transform of $h$.

We here consider the spectral estimation for the exponential Hawkes process, whose reproduction function takes the form $h(t) = \alpha \beta \exp(-\beta t)$, with $\alpha \in (0,1)$ and $\beta > 0$.
Under this specification, the spectral density of the process has the form
\begin{equation}\label{eqn:hawkes_spectrum}
    f(\nu) = m \left( 1 + \frac{\beta^2 \alpha (2 - \alpha)}{\beta^2 (1 - \alpha)^2 + 4\pi^2 \nu^2} \right).
\end{equation}
Proofs that the linear Hawkes process fall under the assumptions in Proposition \ref{prop:contrast} can be found in \cite{Yang2026} (see Section 5.1 therein).

We consider the estimation of $\theta = (\mu, \alpha, \beta)$ using our proposed spectral Ridge estimator, and compare it with more traditional estimation algorithms.
Note that, since we estimate $m = \mu / (1 - \alpha)$ through its plug-in estimator $\widehat m = N_T / T$, the parametric model we consider has the form
\begin{equation*}
    \mathcal{F}_\Theta = \{m\colon 0 < m < \infty\} \times \{ (\alpha, \beta) \colon \alpha \in (0,1), \beta > 0\},
\end{equation*}
and $\mu$ can then be estimated by $\widehat \mu = \widehat m (1 - \widehat \alpha)$.

\subsection{Framework for the simulation study}\label{sec:simulation_framework}

\paragraph{Simulation of the Hawkes process.}
Using the cluster representation of the Hawkes process \citep{Hawkes1974}, we simulated $n_{\text{sim}} = 256$ independent realisations of the exponential Hawkes process on the interval $[0,T]$ with $\mu^* = 1$, $\alpha^* = 0.5$, and $\beta^* = 2$ and a burn-in interval $[-100, 0]$.
Parameters $(\mu, \alpha, \beta)$ are then estimated using a number of different estimators.
For the spectral estimation methods, the periodogram is computed using a Fast Fourier transform \citep{Barnett2019}\footnote{Available under an Apache v2 license.}, and the frequency window, denoted $D$ in Section \ref{sec:contrast}, is fixed at $[-2, 2]$.
This is consistent with the Hawkes spectrum given in Equation \eqref{eqn:hawkes_spectrum}, for which $f_0(\nu) \ll 1$ outside this window for our choice of parameters.

\paragraph{Benchmark estimators.} 
Five different estimators were considered:
\begin{itemize}
    \item Ordinary Least Squares (OLS) estimation for the Hawkes process follows the work of e.g.\ \cite{Reynaud-Bouret2010} and \cite{Bacry2020} wherein the authors define as an estimator the minimiser of a contrast function whose expectation derives from a $L^2$-norm on the space of parameters.
    \item Maximum Likelihood (ML) estimation for the Hawkes process can be traced back to \cite{Ogata1978} and \cite{Ozaki1979}.
    This method is often the favoured one, as it is easy to implement (specifically for the exponential Hawkes process) and enjoys asymptotic statistical efficiency.
    \item Spectral Likelihood (SL) estimation is achieved through the minimisation of a Whittle likelihood. 
    First use case for the Hawkes process can be attributed to \cite{Adamopoulos1976}, though theoretical guarantees are more recent \citep{Cheysson2022, Yang2026}.
    \item Spectral Least Squares (SLS) estimation is our proposed estimation method, detailed in Section \ref{sec:contrast} and given by Equation \eqref{eqn:estimator}.
    \item Spectral Projection (SP) corresponds to the estimator $\widehat f^\perp$ detailed in the remark following Equation \eqref{eqn:estimator}.
\end{itemize}
We implemented the Ridge-penalised form for all five estimation methods, using the $p$-thinning subsampling method when appropriate as detailed in Section \ref{sec:p_thinning_as_subsampling}.
However, since $p$-thinning subsampling techniques cannot be considered for non-spectral methods (as the conditional intensity function appearing in both OLS and ML estimation methods is not tractable for the $p$-thinned Hawkes process), we also considered a different subsampling cross-validation method in order to select the penalisation constant $\kappa$: the observation window $[0,T]$ is partitioned into $k$ equally sized intervals, and the penalisation constant $\kappa$ is selected through Leave-One-Out Cross-Validation (LOOCV), with the performance of the estimation calculated on the left-out interval.
Note that, since the left-out interval usually splits the observation window into two intervals, these two intervals are concatenated such that there is a single training process on an observation window $[0, (k-1)T/k]$, akin to the work of \cite{Reynaud-Bouret2014} (e.g.\ see Figure 2 therein).
A summary recalling the different estimators, with and without Ridge penalisation, can be found in Table \ref{tab:estimators}.

\begin{table}[h]
    \caption{Benchmark Estimators Considered} \label{tab:estimators}
    \begin{center}
        \begin{tabular}{lccc}
        \toprule
        && \multicolumn{2}{c}{With Ridge penalisation}\\
        \cmidrule(lr){3-4}
        & Not penalised & $p$-thinning & LOOCV\\
        \midrule
        OLS & \cmark & \xmark & \cmark\\
        ML & \cmark & \xmark & \cmark\\
        SL & \cmark & \cmark & \cmark\\
        SLS & \cmark & \cmark & \cmark\\
        SP & \cmark & \cmark & \cmark\\
        \bottomrule
        \end{tabular}
    \end{center}
\end{table}

\paragraph{Choice of hyperparameters.} 
For the $p$-thinning-based cross-validation, the number of repeated subsampling was set at $n = 10$.
We considered the following grid for the penalisation parameters $p$ and $\kappa$ for the SLS estimation method:
\[p\in\{0.3, \ldots, 0.8\},\qquad \kappa\in\{2^{-14}, 2^{-13}, \ldots, 2^{3}\}.\]
Note that the choice of the grid for $\kappa$ may depend on the expected value of the objective function, so while this grid proved efficient for the SLS estimation method, we used different grids for $\kappa$ for other methods, adapting them as needed.
Grid values for other estimators may be found in the code in the Supplementary Material\footnote{For anonymity purposes, the code will be released in a GitHub repository and moved out of the Supplementary Material for the camera-ready version.}.
For LOOCV, we chose $k = 4$.

\paragraph{Performance metrics.}\label{sec:performance_metrics}
The different estimators are compared through Mean Square Error (MSE), defined as
\begin{equation*}
    \mathrm{MSE}(\theta) = n_\text{sim}^{-1} \sum_{l=1}^{n_\text{sim}} \lVert \theta_{l} - \theta^*\rVert_2^2.
\end{equation*}

\subsection{Numerical experiments}\label{sec:numerical_experiments}

\paragraph{Main results.}
MSE for all estimators and differing values of $T \in \{50, 100, 200, 400\}$ can be found in Figure \ref{fig:mse}.
Note that, for non-penalised methods, SLS and SP on one hand, and ML and OLS on the other hand, have undistinguishable MSE.

For small values of $T$, penalised methods significantly outperform non-penalised ones.
As displayed by the MSE of each parameter, this is mainly a consequence of non-penalised methods estimating the decay parameter $\beta$ poorly, while first-order parameters $\mu$ and $\alpha$ are broadly well-estimated.
Note also that $p$-thinning-based cross-validation for the SLS fares better than all other spectral approaches considered, highlighting the value of our proposed method.
For large $T$s, the performance of Ridge-penalised methods no longer significantly outperforms that of non-penalised methods, which can be expected from the latter yielding consistent estimators, so that Ridge-penalisation no longer brings improvement asymptotically.

Let us remark that temporal methods (ML and OLS) have lower MSE than spectral methods for all $T$s. 
This is expected for ML for large $T$s, owing to the asymptotic statistical efficiency of maximum likelihood estimation for point processes \citep{Ogata1978}.
Intuitively, it could be explained by spectral approaches only considering the first two moments of the process, while ML and OLS methods are derived from the conditional intensity which contains the whole information of the process.
Regardless, for large $T$s, the slope of the MSE for all methods reaches $-1$ (in log-log scale) corresponding to the $\bigO(T^{-1})$ rate of convergence found in asymptotic results for OLS \citep{Reynaud-Bouret2010}, ML \citep{Ogata1978} and SL \citep{Yang2026}.

\begin{figure*}[h]
    \vspace{.3in}
    \centerline{\includegraphics[width=\linewidth]{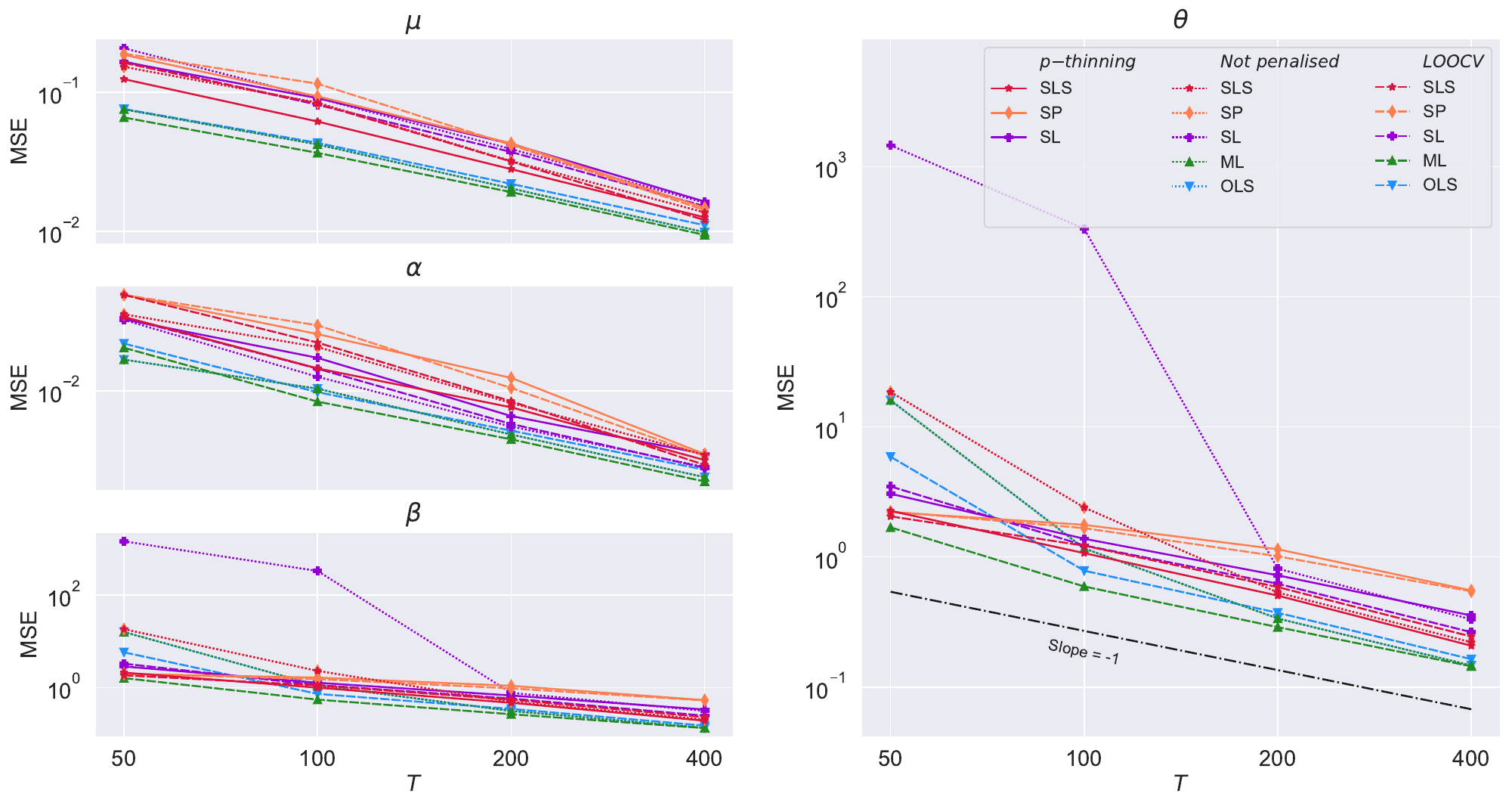}}
    \vspace{.3in}
    \caption{Mean Square Error For Benchmark Estimators}
    \label{fig:mse}
\end{figure*}

An interesting remark can be made about the difference in MSE between SLS and SP.
Even though both estimators are identical for non-penalised estimation as expected (up to numerical error), the story is different when $p$-thinning or LOOCV are involved, as the added term $\int \bigl( \widetilde I_T^{(\bar p)} \bigr)^2(\nu) \dd\nu$ appearing in $\widehat f^\perp$ changes between subsamples.
This creates significant differences when computing the error $r^{(p, \kappa)}$ on the testing sets, hence the difference between SLS and SP when penalising.

\paragraph{Hyperparameter tuning.}
Figure \ref{fig:hyperparameter} shows the number of times each $p$ and $\kappa$ were selected throughout the simulation study for the SLS estimation method.
For all values of $T$, the minimal value for $\kappa$ is selected about half of the time, corresponding to almost no penalisation.
Considering simulations for which $\kappa$ was not selected at its minimal value, then when $T$ increases, the mean value of $\kappa$ decreases, as can be expected from the consistency of the SLS method.

For smaller values of $T$, $p$ is mostly selected equal to 0.8, which is the maximum value explored in the study. As few observed points are available, a compromise must be made in order to keep enough points in the train sample to obtain decent estimations, at the cost of smaller test sets for estimating the testing error. 
As more points become available, we observe that the preferred value of $p$ becomes 0.5, balancing a sizeable training set for estimation and a similarly sized testing set for evaluating the testing error.

\begin{figure*}[h]
    \vspace{.3in}
    \centerline{\includegraphics[width=\linewidth]{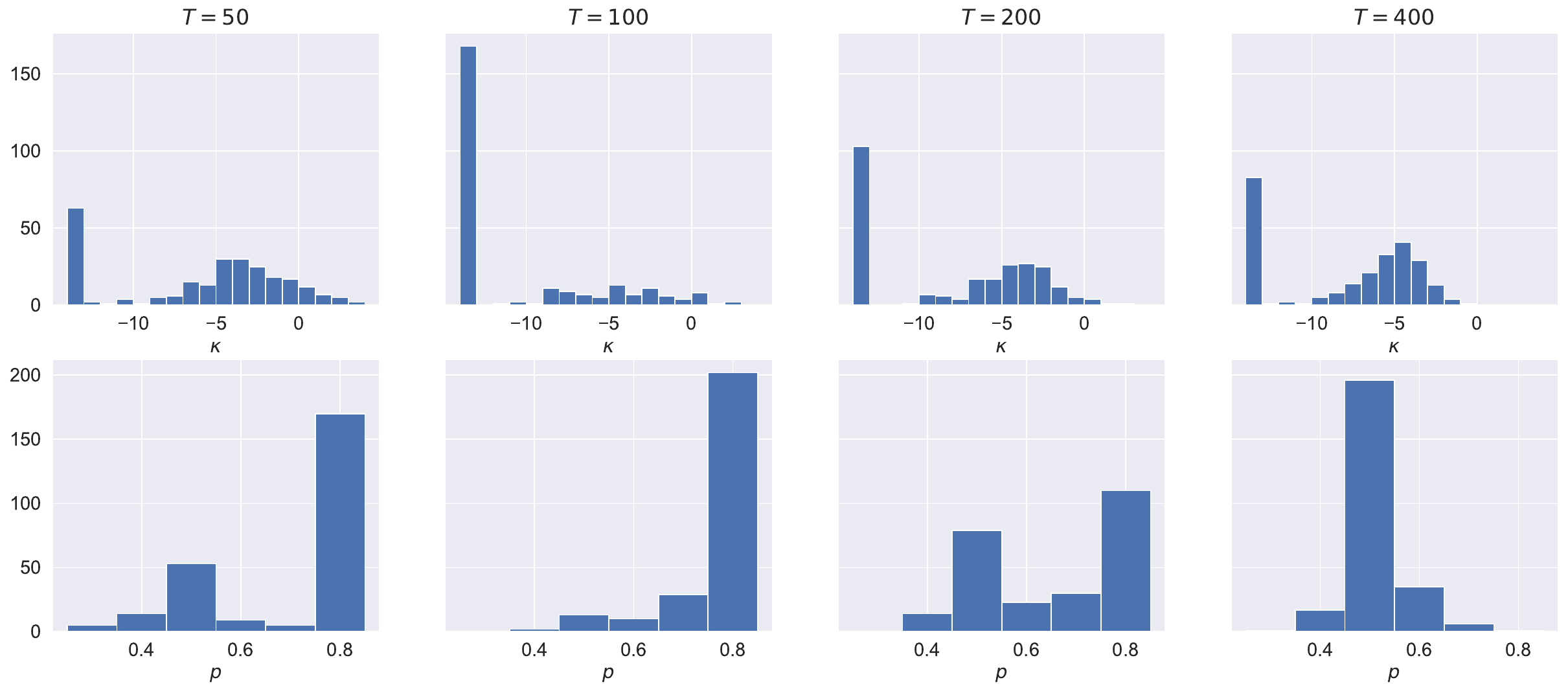}}
    \vspace{.3in}
    \caption{Proportion Of Simulations For Which $\kappa$ (In Base-2 Logarithm) And $p$ Are Selected For The SLS}
    \label{fig:hyperparameter}
\end{figure*}

\paragraph{Computation time.}
Average time (in seconds) for the computation of each studied estimator on a realisation of the Hawkes process on the interval $[0, 100]$ can be found in Table \ref{tab:computation_time}, carried on a MacBook Pro with Apple M4 chip and 24 GB RAM.
As expected, the Ridge-penalised estimation methods took orders of magnitude more time than non-penalised methods.
Note also that, due to the choice of selecting both $p$ and $\kappa$ in $p$-thinning-based cross-validation rather than just $\kappa$ in LOOCV, the former method takes a lot more time than the latter.
This could be circumvented by choosing an arbitrary value for $p$ (as we did for $k$ in LOOCV), bringing in this case the computation time of both cross-validation techniques to the same order.

\begin{table}[h]
    \caption{Computation Time For Each Benchmark Estimators (In Seconds)} \label{tab:computation_time}
    \begin{center}
        \begin{tabular}{lccc}
        \toprule
        && \multicolumn{2}{c}{With Ridge penalisation}\\
        \cmidrule(lr){3-4}
        & Not penalised & $p$-thinning & LOOCV\\
        \midrule
        OLS & 0.13 & \xmark & 87.10\\
        ML & 0.13 & \xmark & 18.25\\
        SL & 0.25 & 24.85 & 1.36\\
        SLS & 0.13 & 11.24 & 0.58\\
        SP & 0.14 & 40.41 & 2.14\\
        \bottomrule
        \end{tabular}
    \end{center}
\end{table}

Let us also remark that the choice of the exponential reproduction function is not insignificant and has an influence on computation time: since the exponential Hawkes process exhibits Markovian properties, its conditional intensity function can be computed in $\bigO(T)$ steps rather than the usual $\bigO(T^2)$ steps needed \citep{Ozaki1979}.
In contrast, the main computational burden in computing the objective function in spectral methods comes from the estimation of the periodogram of the process, which takes $\bigO(T \log T)$ steps using a fast Fourier transform.
Hence another choice of reproduction function would not shine so clearly on temporal (ML and OLS) methods, for which spectral methods would have a significant edge in computation time.

\section{Conclusion}
This article introduces a Ridge-penalised spectral least-squares estimator for second-order stationary point processes based around a $p$-thinning-based cross-validation method for hyperparameter tuning.
Our simulation study shows that this method yields better performances than existing spectral approaches against a single short realisation of the process, which may prove useful when dealing with real data. This is specially advantageous in contexts where conventional estimation methods are unavailable (e.g.\ likelihood methods when the Papangelou conditional intensity of the process is untractable).

\bibliography{bibliography}
\bibliographystyle{abbrvnatnourl}

\appendix

\section{PROOF OF PROPOSITION \ref{prop:spectral_thinning}}\label{app:proof_spectral_thinning}

Let us introduce an alternative way of viewing the $p$-thinning of a point process through marked point process theory (see also \cite{Cronie2024}).
Define the marked point process $\bar N$ associated with $N$, with marks $(Z_k)_{t\in\Z}$ on a metric space $\mathcal{K}$, as the collection of points $(T_k, Z_k)_{k\in\Z} \in (\R \times \mathcal{K})^{\Z}$ (see \cite[Chapter 6.4]{DaleyV1} for a more thorough presentation of marked point processes).
The random marks $Z_k$ are usually used to represent underlying information on the event times of a point process $N$, which is often referred to as the \textit{ground} process.
We will restrict ourselves to the case where the random variables $(Z_k)_{k\in\Z}$ are independent and identically distributed.
In this setting, process $\bar N$ is well-defined \citep[6.4.IV(a)]{DaleyV1}.

We can then view a $p$-thinning of $N$ as a marked version $\bar N$ where $\mathcal{K} = \{0,1\}$
and the $(Z_k)_{k\in\Z}$ are a collection of Bernoulli random variables of parameter $p$.
This way we may define the thinned process $N_p$, for any $B\in\mathcal{B}(\R)$, as:
\[N_p(B) = \bar N(B\times\{1\})\,.\]
Under this scope, we recall the results of \cite{Bremaud2005}, that we adapt to our notations.
The utility of these results lies on the link that it establishes between the covariance of a marked point process $\bar N$ and the Bartlett spectrum of its ground process $N$.

\begin{theorem}[{\cite[Theorem 2]{Bremaud2005}}]\label{th:spectral_marked}
    Let $N$ be a stationary point process with mean intensity $m$ and spectral measure $\Gamma$ and $\bar N$ a marked version of $N$ with i.i.d.\ marks $Z_k$ with shared distribution $Z$ on a metric space $\mathcal{K}$.
    Let $\varphi^\star, \psi^\star$ be measurable functions from $\R \times \mathcal{K} \to \R$, 
    such that:
    \begin{itemize}
        \item $\displaystyle
            \int_{\R}{\E\left[|\varphi^\star(x, Z)|\right]\dd x} < +\infty\,,\qquad\qquad \int_{\R}{\E\left[|\psi^\star(x, Z)|\right]\dd x} < +\infty\,.
        $
        \item $\displaystyle
            \int_{\R}{\E\left[\varphi^\star(x, Z)^2\right]\dd x} < +\infty\,,\qquad\qquad \int_{\R}{\E\left[\psi^\star(x, Z)^2\right]\dd x} < +\infty\,.
        $
        \item Denoting $\bar \varphi : x\to \E\left[\varphi^\star(x, Z)\right]$ and $\bar \psi : x\to \E\left[\psi^\star(x, Z)\right]$, assume that $\bar \varphi,\bar \psi \in L^2(\R) \cap L^1(\R)$.
    \end{itemize}

    Then, it follows that:
    \begin{align}\label{eq:marked_spectrum}
        \Cov\left(\sum_{k\in\Z}{\varphi^\star(T_k, Z_k)}, \sum_{k\in\Z}{\psi^\star(T_k, Z_k)}\right) =& 
        \int_{\R}{\widetilde {\bar\varphi} (\nu) \widetilde{\bar \psi} (-\nu)\,\Gamma(\dd \nu)} \\
        &+ \int_{\R}{\Cov\left( \widetilde \varphi^\star(\nu, Z), \widetilde \psi^\star (-\nu, Z)\right) m \dd \nu}\nonumber\,,
    \end{align}
    where, for any $\nu\in\R$ and function $f \in L^2(\R)$ (resp.\ $f \in L^2(\R \times \mathcal K)$), $\widetilde f(\nu)$ (resp.\ $\widetilde f(\nu, Z)$) denotes the Fourier transform of $x \mapsto f(x)$ (resp.\ $x \mapsto f(x, Z)$).
\end{theorem}

Within the notations of Theorem \ref{th:spectral_marked}, consider that the $Z_k$ are i.i.d.\ Bernoulli random variables with common probability $p$, and let $\varphi, \psi \in L^2(\R) \cap L^1(\R)$.
We define, for all $x,z\in\R\times\{0,1\}$, the functions
\[\varphi^\star(x,z) = \varphi(x)z\,,\qquad \psi^\star(x,z) = \psi(x)z\,.\]
Let us verify that these functions satisfy the conditions of Theorem~\ref{th:spectral_marked}. 
Without loss of generality, we will work uniquely with $\varphi^\star$, as the arguments are exactly the same for $\psi^\star$.

For any $x\in\R$, $\varphi^\star(x, Z) = \varphi(x)Z$ for Z a Bernoulli random variable with parameter $p$.
It follows that $\varphi^\star(x, Z)$ admits a first- and second-order moment, and as $\varphi\in L^2(\R) \cap L^1(\R)$, 
$\varphi(x, Z)$ is integrable and square integrable, which shows that:
        \[
        \int_{\R}{\E\left[|\varphi^\star(x, Z)|\right]\dd x} < +\infty\,,\qquad 
        \int_{\R}{\E\left[\varphi^\star(x, Z)^2\right]\dd x} < +\infty\,.
        \]
Furthermore, for any $x\in\R$, $\bar \varphi(x) = \varphi(x)\E[Z] = p \varphi(x)$,
and so, as $L^2(\R) \cap L^1(\R)$ is closed under scalar multiplication, it follows that $\bar \varphi\in L^2(\R) \cap L^1(\R)$.

We can then apply Equation~\eqref{eq:marked_spectrum} to our marked process $\bar N$. 
For this, let us notice that:
\[
    \sum_{k\in\Z}{\varphi^\star(T_k, Z_k)} = \sum_{k\in\Z}{\varphi(T_k)Z_k} = \int_{\R}{\varphi(t) N_p(\dd t)}\,,
\]
with the same expression holding for $\psi^\star$ and $\psi$. 
So, the left-hand side of Equation~\eqref{eq:marked_spectrum} reads:
\begin{align}
    \Cov\left(\sum_{k\in\Z}{\varphi^\star(T_k, Z_k)}, \sum_{k\in\Z}{\psi^\star(T_k, Z_k)}\right) &=
    \Cov\left(\int_{\R}{\varphi(t) N_p(\dd t)}, \int_{\R}{\psi(t) N_p(\dd t)}\right) \nonumber\\
    &= \int_{\R}{\widetilde \varphi(\nu)\widetilde \psi(-\nu)\,\Gamma_p(\dd\nu)}\,,\label{eq:leftside_proof_thinning}
\end{align}
where the last equality comes from polarising Equation~\eqref{eq:bartlett_variance}.

For the right-hand side, let us remark that $\bar \varphi(x) = p \varphi(x)$ for any $x\in\R$ and so, for any $\nu\in\R$,
\[
    \widetilde {\bar \varphi}(\nu) = p \widetilde \varphi(\nu), \qquad \text{and} \qquad \widetilde \varphi^\star(\nu, Z) = \widetilde \varphi(\nu)Z.
\]
The right-hand side of Equation~\eqref{eq:marked_spectrum} then becomes:
\begin{align}
    \int_{\R}{\widetilde {\bar\varphi} (\nu) \widetilde{\bar \psi} (-\nu)\,\Gamma(\dd \nu)} 
+ &\int_{\R}{\Cov\left( \widetilde \varphi^\star(\nu, Z), \widetilde \psi^\star (-\nu, Z)\right)m \dd \nu} \nonumber\\
&= \int_{\R}{p^2 \widetilde {\varphi} (\nu) \widetilde{\psi} (-\nu)\,\Gamma(\dd \nu)} 
+ \int_{\R}{\widetilde \varphi(\nu)\widetilde \psi(\nu) \Cov\left(Z,Z\right) m \dd \nu} \nonumber\\
&= \int_{\R}{\widetilde {\varphi} (\nu) \widetilde{\psi} (-\nu)\,\bigl( p^2\Gamma(\dd \nu) + p(1-p) m \dd \nu \bigr)}\,.\label{eq:rightside_proof_thinning}
\end{align}

Combining both sides (Equations~\eqref{eq:leftside_proof_thinning} and \eqref{eq:rightside_proof_thinning}), it follows that
for any $\varphi, \psi\in L^2(\R) \cap L^1(\R)$:
\[
    \int_{\R}{\widetilde \varphi(\nu)\widetilde \psi(-\nu)\,\Gamma_p(\dd\nu)} = \int_{\R}{\widetilde {\varphi} (\nu) \widetilde{\psi} (-\nu)\, \bigl(p^2\Gamma(\dd \nu) + p(1-p) m \dd \nu \bigr)}.
\]
As this equality holds for any functions in $L^2(\R) \cap L^1(\R)$ so that, by duality of the Fourier transform \citep{Pinsky2008},
\[\Gamma_p = p^2 \Gamma + p(1-p) m,\]
which achieves the proof.
\end{document}